\documentclass[preprint,1p,11pt]{IR-Template/ISAS_IR}

\usepackage[cmex10]{amsmath}
\usepackage{amsfonts}
\usepackage[ruled]{algorithm2e}
\usepackage[english]{babel}
\usepackage{ISASPackages/ISASMatheMacros}
\usepackage{graphicx}
\usepackage{epstopdf}
\usepackage{flushend}
\usepackage{siunitx}
\usepackage{booktabs}

\usepackage{amsthm}

\newtheorem{lemma}{Lemma}
\newtheorem{proposition}{Proposition}

\usepackage{hyperref}

\begin{document}

\begin{frontmatter}

\title{Efficient Evaluation of the Probability Density Function of a Wrapped Normal Distribution}

\author[isas]{Gerhard~Kurz}
\ead{gerhard.kurz@kit.edu}

\author[isas]{Igor~Gilitschenski}
\ead{gilitschenski@kit.edu}

\author[isas]{Uwe~D.~Hanebeck}
\ead{uwe.hanebeck@ieee.org}

\address[isas]{Intelligent Sensor-Actuator-Systems Laboratory (ISAS)\\
Institute for Anthropomatics and Robotics\\
Karlsruhe Institute of Technology (KIT), Germany\vspace{3mm}}

\begin{abstract}
The wrapped normal distribution arises when a the density of a one-dimensional normal distribution is wrapped around the circle infinitely many times. At first look, evaluation of its probability density function appears tedious as an infinite series is involved. In this paper, we investigate the evaluation of two truncated series representations. As one representation performs well for small uncertainties whereas the other performs well for large uncertainties, we show that in all cases a small number of summands is sufficient to achieve high accuracy.
\end{abstract}

\end{frontmatter}

\section{Introduction} \label{sec:introduction}
The wrapped normal (WN) distribution is one of the most widely used distributions in circular statistics. Applications for the WN distribution include circular filtering \cite{ACC13_Kurz}, \cite{ACC14_Kurz}, constrained tracking \cite{IPIN13_Kurz}, speech processing \cite{agiomyrgiannakis2009}, and bearings-only tracking \cite{Fusion13_Gilitschenski}. However, evaluation of the WN probability density function can appear difficult because it involves an infinite series. This is one of the main reasons why many authors (such as \cite{azmani2009}, \cite{stienne2013}, \cite{mokhtari2013}) use the von Mises distribution instead, which is even sometimes referred to as circular normal distribution \cite{jammalamadaka2001}. In this paper, we will show that a very accurate numerical evaluation of the WN probability density function can be performed with little effort.

The wrapped normal distribution \cite{jammalamadaka2001}, \cite{mardia1999} is defined by the probability density function (pdf)
\begin{align*}
f(x;\mu,\sigma) = \frac{1}{\sqrt{2 \pi} \sigma} \sum_{k=-\infty}^\infty \exp \left( - \frac{(x + 2\pi k - \mu)^2}{2 \sigma^2}  \right) \ , 
\end{align*}
with $x \in [0, 2 \pi)$, location parameter $\mu \in [0, 2 \pi)$, and uncertainty parameter $\sigma > 0$. Because the summands of the series converge to zero, it is natural to approximate the pdf with a truncated series
\begin{align*}
f(x;\mu,\sigma) \approx f_n(x;\mu,\sigma) = \frac{1}{\sqrt{2 \pi} \sigma} \sum_{k=-n}^n \exp \left( - \frac{(x + 2\pi k - \mu)^2}{2 \sigma^2}  \right) \ ,
\end{align*}
where only $2n+1$ summands are considered. We will investigate the choice of $n$ (depending on $\sigma$) in this paper.

As we will later prove, the series representation defined above yields a good approximation for small values of $\sigma$ only. For this reason, we introduce a second representation, which yields good approximations for large values of $\sigma$. The pdf of a WN distribution is closely related to the Jacobi theta function \cite{abramowitz1972}. This leads to another representation of the pdf \cite[(2.2.15)]{jammalamadaka2001}
\begin{align*}
g(x;\mu,\sigma) = \frac{1}{2 \pi} \left( 1 + 2 \sum_{k=1}^\infty \rho^{k^2} \cos (k (x-\mu)) \right) \ ,
\end{align*}
where $\rho = \exp( -\sigma^2/2)$ . Analogous to $f_n$, we define a truncated version 
\begin{align*}
g(x;\mu,\sigma) \approx g_n(x;\mu,\sigma) = \frac{1}{2 \pi} \left( 1 + 2 \sum_{k=1}^n \rho^{k^2} \cos (k (x-\mu)) \right) \ ,
\end{align*}
which only considers the first $n$ summands.\footnote{We treat the parameter $n$ in $f_n$ and $g_n$ the same way, although the evaluation of $f_n$ involves $2n+1$ summands whereas the evaluation of $g_n$ only involves $n$ summands. However, the computational effort for evaluation of a single summand of $g_n$ is higher, which roughly negates this difference.}

\section{Empirical Results}
We implemented the truncated series $f_n$ and $g_n$ as well as the exact solution (which increases $n$ until the value of the pdf does not change anymore because of the limited accuracy of the data type). We used the IEEE 754 double data type for all variables. It consists of 1 bit for the sign, 11 bit for the  exponent, and 52 bit for the fraction \cite{goldberg1991}. Thus, it is accurate to approximately 15 decimal digits. 

\begin{figure}[h]
	\centering
	\includegraphics[width=8cm]{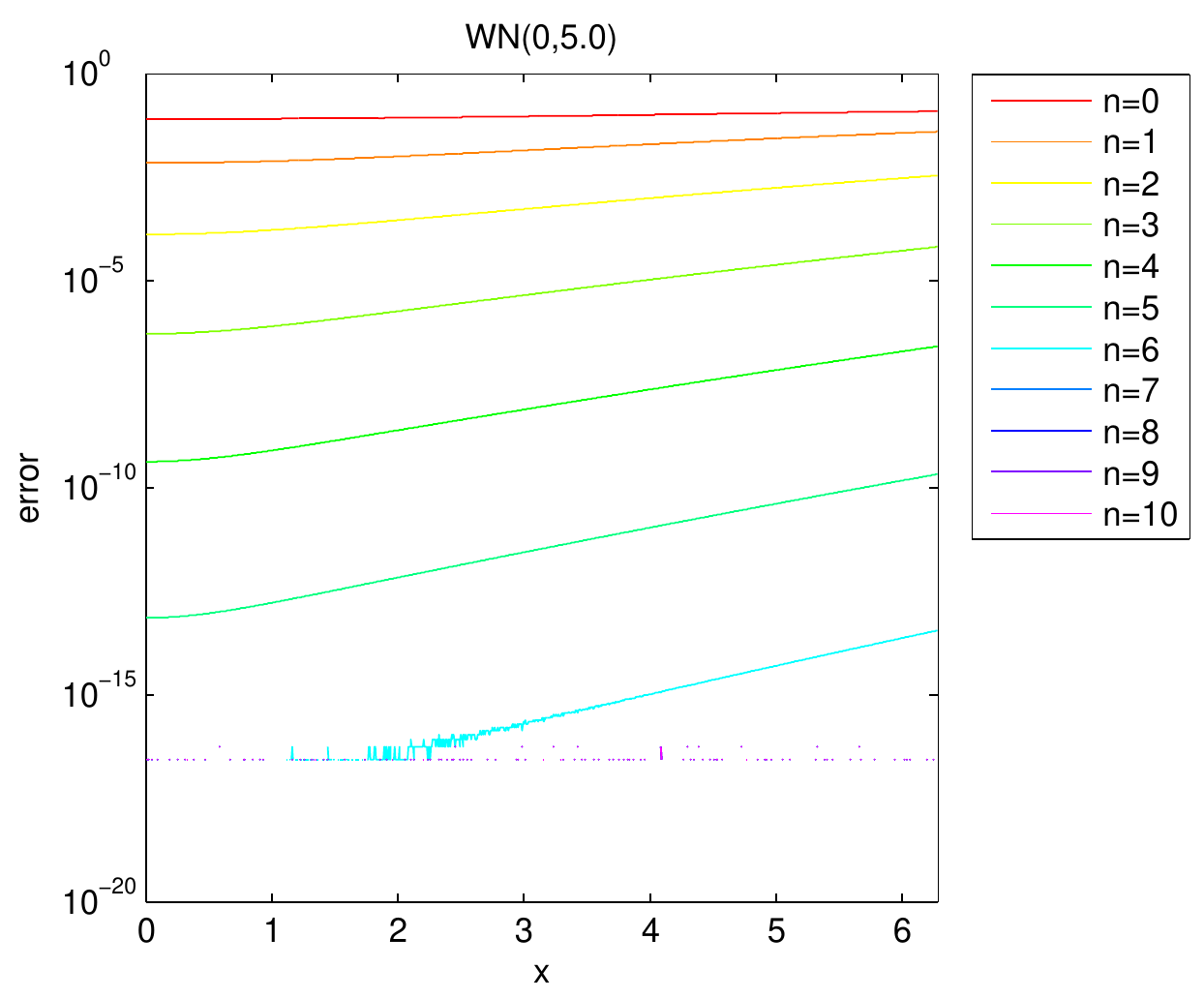}\quad  \includegraphics[width=8cm]{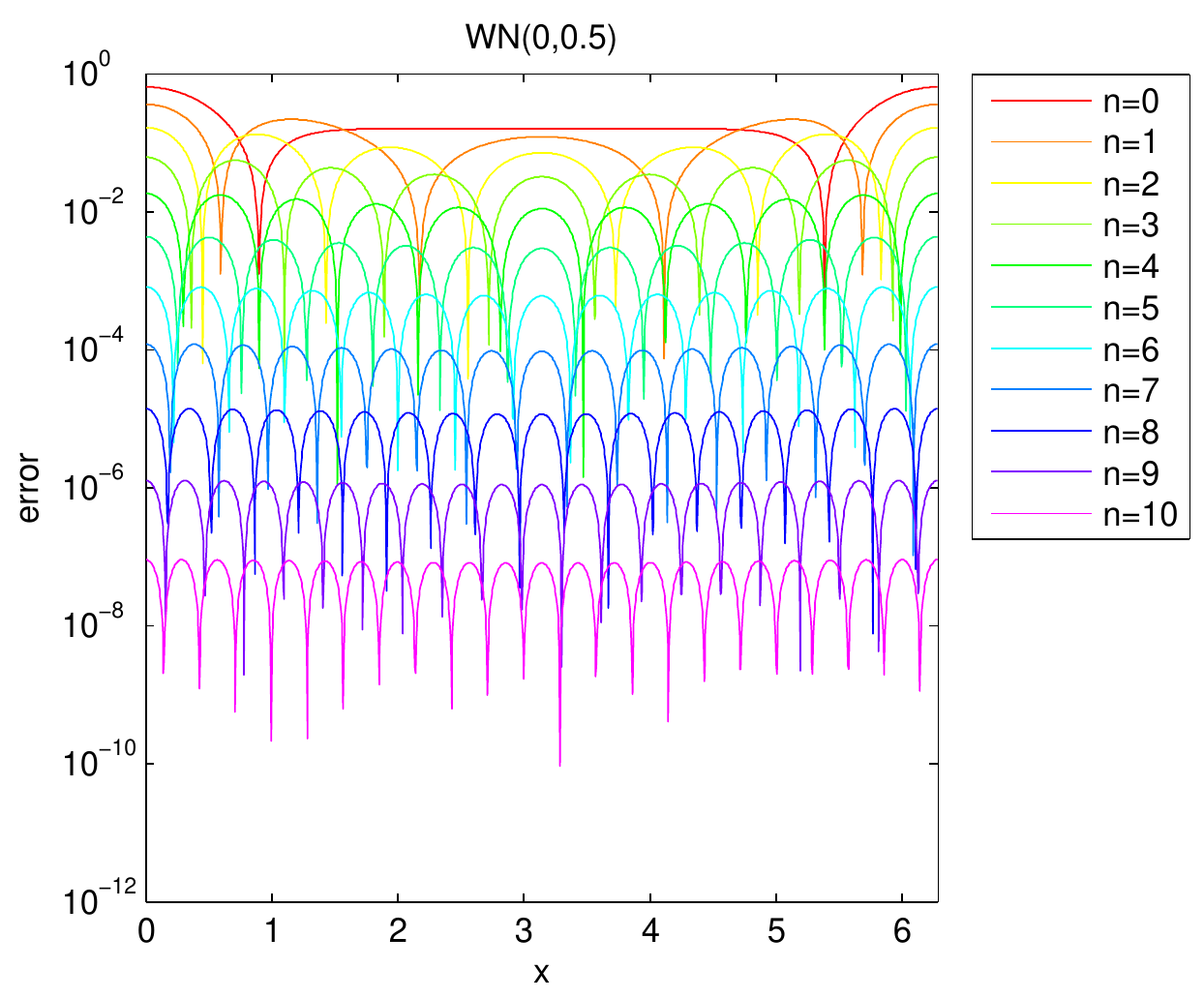} 
	\caption{Empirical results depicting the error for different values of $n$ for $e_f(n,\sigma)$ with $\sigma=5$ (left) and $e_g(n,\sigma)$ with $\sigma = 0.5$ (right). Note that some points are rounded to zero because of the limited accuracy of the floating point arithmetic. These values are not depicted, because it is not possible to display them in a logarithmic plot.}
	\label{fig:empiricalresults-x}
\end{figure}

For $x, \mu \in [0,2 \pi)$, the error is largest for $\mu=0$ and $x \to 2 \pi$ in both approximations (see Fig.~\ref{fig:empiricalresults-x}). We will later show this fact in the theoretical section. Thus, we compare the error $e_f(n, \sigma) = |f(2 \pi;0;\sigma)-f_n(2 \pi,0,\sigma)|$ and $e_g(n,\sigma) = |g(2 \pi;0;\sigma)-g_n(2 \pi,0,\sigma)|$ respectively. The results for $n=1,2, \dots, 11$ are depicted in Fig.~\ref{fig:empiricalresults}. Furthermore, we include a comparison to the uniform distribution with pdf $f_u(x) = \frac{1}{2 \pi}$, which is also a special case of $g_n$ for $n=0$.
\begin{figure}[h]
	\centering
	\includegraphics[width=8cm]{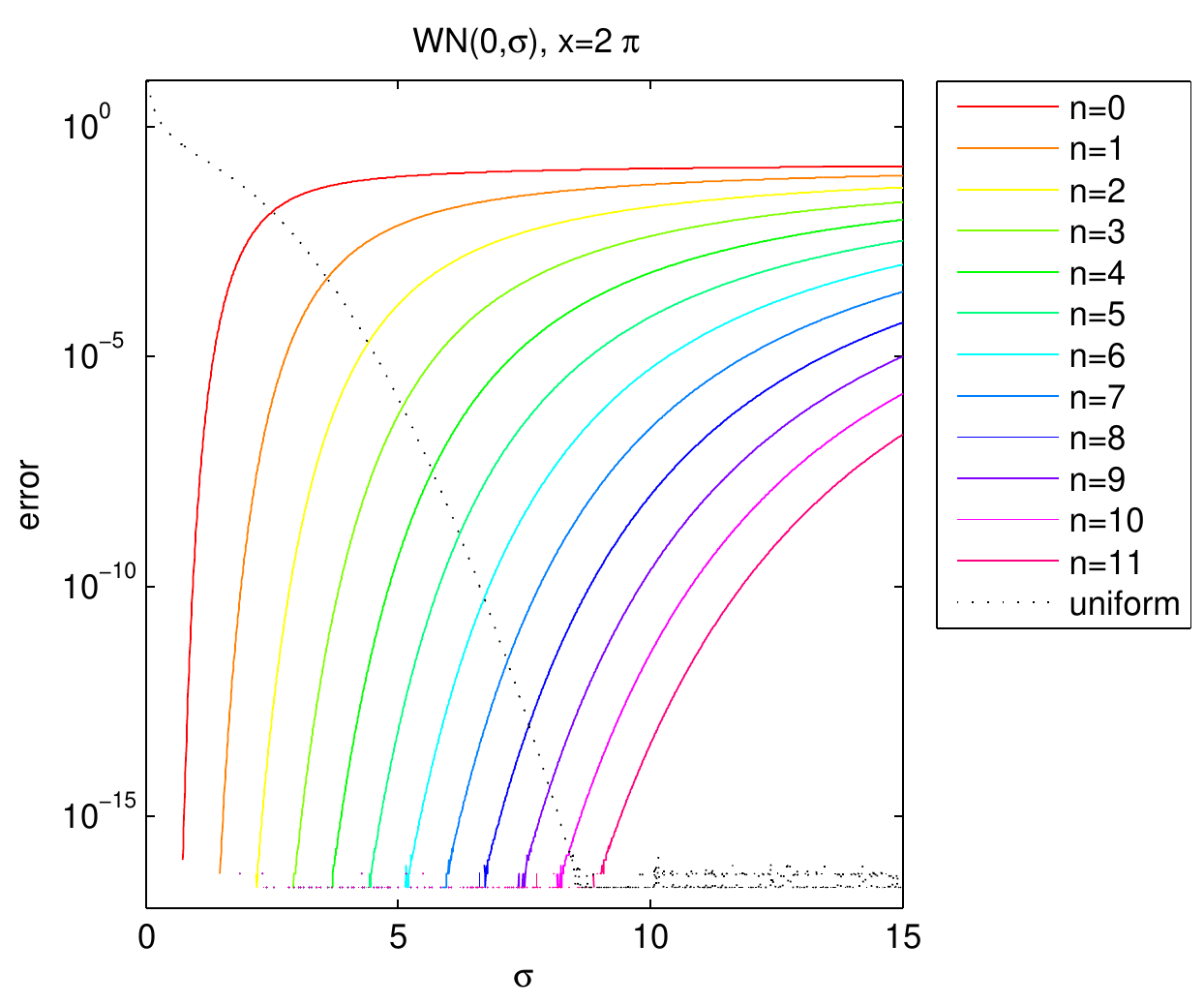}\quad \includegraphics[width=8cm]{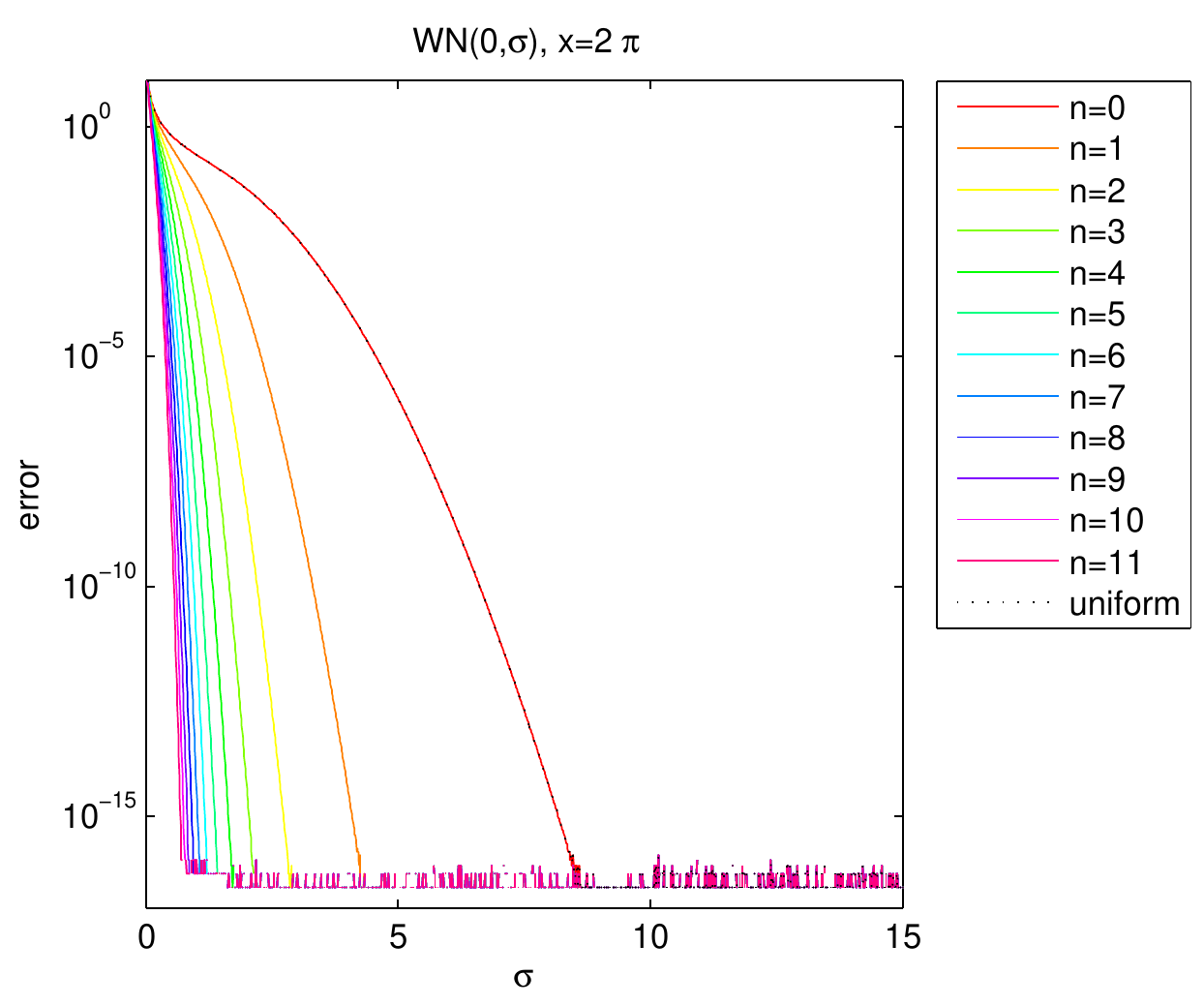} 
	\caption{Empirical results depicting the error for different values of $n$ for $e_f(n,\sigma)$ (left) and $e_g(n,\sigma)$ (right). We set the WN parameter $\mu=0$ and $x= 2\pi$.}
	\label{fig:empiricalresults}
\end{figure}
As can be seen, the uniform distribution is accurate up to numerical precision for approximately $\sigma\geq 9$.

We empirically determined the combined approximation based on $f_n$ and $g_n$ for different accuracies (see Table~\ref{table:combineapprox}).


\begin{table}
	\centering
	\begin{tabular}{ccc}
	\toprule
	\bf{accuracy} & \bf{range} & \bf{approximation} \\
	\midrule
	& $0 < \sigma < 1.34$ & $f^0(x; \mu, \sigma)$ \\
	1E-5 & $1.34 \leq \sigma < 2.28$ & $f^1(x; \mu, \sigma)$ \\
	& $2.28 \leq \sigma < 4.56$ & $g^1(x; \mu, \sigma)$ \\
	& $ 4.56 \leq \sigma$ & $g^0(x; \mu, \sigma)$ \\
	\midrule
	& $0 < \sigma < 0.93$ & $f^0(x; \mu, \sigma)$ \\
	& $0.93 \leq \sigma < 1.89$ & $f^1(x; \mu, \sigma)$ \\
	1E-10 & $1.89 \leq \sigma < 2.21$ & $f^2(x; \mu, \sigma)$ \\
	& $2.21 \leq \sigma < 3.31$ & $g^2(x; \mu, \sigma)$ \\
	& $3.31 \leq \sigma < 6.62$ & $g^1(x; \mu, \sigma)$ \\
	& $6.62 \leq \sigma$ & $g^0(x; \mu, \sigma)$ \\
	\midrule
	& $0 < \sigma < 0.76$ & $f^0(x; \mu, \sigma)$ \\
	& $0.76 \leq \sigma < 1.53$ & $f^1(x; \mu, \sigma)$ \\
	& $1.53 \leq \sigma < 2.31$ & $f^2(x; \mu, \sigma)$ \\
	1E-15 & $2.31 \leq \sigma < 2.73$ & $g^3(x; \mu, \sigma)$ \\
	& $2.73 \leq \sigma < 4.09$ & $g^2(x; \mu, \sigma)$ \\
	& $4.09 \leq \sigma < 8.17$ & $g^1(x; \mu, \sigma)$ \\
	& $8.17 \leq \sigma$  & $g^1(x; \mu, \sigma)$ \\
	\bottomrule
	\end{tabular}
	\caption{Combined approximations for different accuracies.}
	\label{table:combineapprox}
\end{table}

\section{Theoretical Results}
Before we analyze the approximation error of the different approaches, we prove an inequality for the error function.
\begin{lemma}
	\label{lemma:erfinequality}
	For $x>1$, the error function fulfills the inequality $ 1- \operatorname{erf}(x) \leq \frac{e^{-x^2}}{\sqrt{\pi}}$.
\end{lemma}
\begin{proof}
We use the continued fraction representation \cite[7.1.14]{abramowitz1972}
\begin{align*}
\operatorname{erf}(x) &= 1 - \frac{e^{-x^2}}{\sqrt\pi\left(x + \frac 1{2x + \frac 2{x + \frac 3{2x + \frac 4{x + \dotsb}}}}\right)} \\
\Rightarrow 1- \operatorname{erf}(x) &= \frac{e^{-x^2}}{\sqrt\pi\left(x + \frac 1{2x + \frac 2{x + \frac 3{2x + \frac 4{x + \dotsb}}}}\right)} \\
\underset{x>1}{\Rightarrow} 1- \operatorname{erf}(x) &\leq \frac{e^{-x^2}}{\sqrt{\pi}}
\end{align*} 
\end{proof}

\subsection{Representation Based on Wrapped Density}
We consider the approximation $f_n(x;\mu, \sigma) \approx f(x;\mu, \sigma)$. In the following proposition, we will show that the error decreases exponentially in $n$.

\begin{proposition}
	\label{prop:f}
	For $x, \mu \in [0, 2 \pi)$ and $n > 1+ \frac{\sigma}{\sqrt{2} \pi}$, the error $e_f(n, \sigma) = \left| f_n (x ; \mu, \sigma) - f (x ; \mu, \sigma) \right|$ has an upper bound
	\begin{align*}
	e_f(n, \sigma) < \frac{ \exp \left(- \frac{(\pi \sqrt{2} (n-1) )^2}{\sigma^2} \right) }{2 \pi^{3/2}} \ .
	\end{align*}
\end{proposition}
\begin{proof}
We use the fact that $\sigma>0$ and $\exp(\cdot)>0$, and get
\begin{align}
e_f(n, \sigma) &= \left| f_n (x ; \mu, \sigma) - f (x ; \mu, \sigma) \right| \nonumber \\
&= \left| \frac{1}{\sqrt{2 \pi} \sigma} \sum_{k=-n}^n \exp \left( - \frac{(x - \mu - 2 k \pi )^2}{2 \sigma^2}  \right) - \frac{1}{\sqrt{2 \pi} \sigma} \sum_{k=-\infty}^\infty \exp \left( - \frac{(x - \mu - 2 k \pi )^2}{2 \sigma^2}  \right) \right| \nonumber \\
&\underset{\sigma > 0}{=} \frac{1}{\sqrt{2 \pi} \sigma} \left| \sum_{k=-\infty}^{-n-1} \exp \left( - \frac{(x - \mu - 2 k \pi )^2}{2 \sigma^2}  \right) + \sum_{k=n+1}^\infty \exp \left( - \frac{(x - \mu - 2 k \pi )^2}{2 \sigma^2}  \right) \right| \nonumber \\ 
&\underset{\exp(\cdot)>0}{=} \frac{1}{\sqrt{2 \pi} \sigma} \left( \sum_{k=-\infty}^{-n-1} \exp \left( - \frac{(x - \mu - 2 k \pi )^2}{2 \sigma^2}  \right) + \sum_{k=n+1}^\infty \exp \left( - \frac{(x - \mu - 2 k \pi )^2}{2 \sigma^2}  \right) \right) \ . \label{eq:fstep1} \\
\intertext{Now we make use of the fact that $\mu$ and $x$ are in the same interval of length $2\pi$, and combine the two series into one}
(\ref{eq:fstep1}) &\underset{|x-\mu|<2\pi }{<} \frac{1}{\sqrt{2 \pi} \sigma} \left( \sum_{k=-\infty}^{-n-1} \exp \left( - \frac{(- 2 \pi - 2 k \pi )^2}{2 \sigma^2}  \right) + \sum_{k=n+1}^\infty \exp \left( - \frac{( 2 \pi  - 2 k \pi )^2}{2 \sigma^2}  \right) \right) \nonumber \\
&= \frac{1}{\sqrt{2 \pi} \sigma} \left( \sum_{k=-\infty}^{-n-1} \exp \left( - \frac{(- 2 (k+1) \pi )^2}{2 \sigma^2}  \right) + \sum_{k=n+1}^\infty \exp \left( - \frac{( - 2 (k-1) \pi )^2}{2 \sigma^2}  \right) \right) \nonumber \\
&= \frac{1}{\sqrt{2 \pi} \sigma} \left( \sum_{k=-\infty}^{-n} \exp \left( - \frac{(2 k \pi )^2}{2 \sigma^2}  \right) + \sum_{k=n}^\infty \exp \left( - \frac{(2 k \pi )^2}{2 \sigma^2}  \right) \right) \nonumber \\
&= \frac{2}{\sqrt{2 \pi} \sigma} \sum_{k=n}^{\infty} \exp \left( - \frac{(2 k \pi )^2}{2 \sigma^2}  \right) \ , \label{eq:fstep2} \\
\intertext{and find an upper bound by integration}
(\ref{eq:fstep2}) &\leq \frac{2}{\sqrt{2 \pi} \sigma} \int_{k=n-1}^{\infty} \exp \left( - \frac{(2 k \pi )^2}{2 \sigma^2}  \right) dk \nonumber \\
&\underset{\text{\cite[7.1.2]{abramowitz1972}}}{=} \frac{\left(1 - \erf \left(\frac{\pi\, \sqrt{2}\, \left(n - 1\right)}{\sigma}\right) \right)}{2 \pi} \nonumber \\
&\underset{\text{Lemma \ref{lemma:erfinequality}}}{\leq} \frac{ \exp \left(- \frac{(\pi \sqrt{2} (n-1) )^2}{\sigma^2} \right)}{2 \pi^{3/2}} \nonumber \ ,
\end{align}
where we use the assumption $\frac{\pi\, \sqrt{2}\, \left(n - 1\right)}{\sigma}>1$ in order to apply Lemma~\ref{lemma:erfinequality}.
\end{proof}

\subsection{Representation Based on Theta Function}
In the following, we consider the approximation $g_n(x;\mu, \sigma) \approx g(x;\mu, \sigma)$. In this case, the error decreases exponentially in $n$ as well.

\begin{proposition}
	\label{prop:g}
	For $x, \mu \in [0, 2 \pi)$ and $n > \sqrt{2}/\sigma$, the error $e_g(n, \sigma) = \left| g_n (x ; \mu, \sigma) - g (x ; \mu, \sigma) \right|$ has an upper bound
	\begin{align*}
	e_g(n, \sigma) <\frac{\exp( - n^2 \sigma^2 / 2) }{\sqrt{2} \pi \sigma} \ .
	\end{align*}
\end{proposition}
\begin{proof}
We start with some simplifications
\begin{align}
e_g(n, \sigma) &= \left| g_n (x ; \mu, \sigma) - g (x ; \mu, \sigma) \right| \nonumber \\
&= \left|\frac{1}{2 \pi} \left( 1 + 2 \sum_{k=1}^n \rho^{k^2} \cos (k (x-\mu)) \right) - \frac{1}{2 \pi} \left( 1 + 2 \sum_{k=1}^\infty \rho^{k^2} \cos (k (x-\mu)) \right) \right| \nonumber \\
&=  \frac{1}{\pi} \left| \sum_{k=1}^n \rho^{k^2} \cos (k (x-\mu)) - \sum_{k=1}^\infty \rho^{k^2} \cos (k (x-\mu)) \right| \nonumber \\
&=  \frac{1}{\pi} \left| \sum_{k={n+1}}^\infty \rho^{k^2} \cos (k (x-\mu)) \right| \ , \label{eq:gstep1} \\
\intertext{use the triangle inequality and the fact that $|\cos(\cdot)| \leq 1$}
(\ref{eq:gstep1}) &\leq  \frac{1}{\pi} \sum_{k={n+1}}^\infty \rho^{k^2} \ .  \label{eq:gstep2} \\
\intertext{Now we find an upper bound by integration and simplify}
(\ref{eq:gstep2}) &\leq  \frac{1}{\pi} \int_{n}^\infty \rho^{k^2} dk \nonumber \\
&\underset{\text{\cite[7.1.2]{abramowitz1972}}}{=} \frac{1}{\pi} \cdot \frac{\sqrt{\pi} \text{ erfc}(n \sqrt{-\log(\rho)}) }{2 \sqrt(-\log(\rho))} \nonumber \\
&= \frac{1}{\pi} \cdot \frac{\sqrt{\pi} \text{ erfc}(n \sqrt{\sigma^2/2}) }{2 \sqrt{\sigma^2/2}} \nonumber \\
&= \frac{1 - \erf (n \sigma / \sqrt{2}) }{\sqrt{2 \pi} \sigma} \nonumber \\
&\underset{\text{Lemma \ref{lemma:erfinequality}}}{\leq} \frac{\exp( - n^2 \sigma^2 / 2) }{\sqrt{2} \pi \sigma} \ , \nonumber
\end{align}
where we use the assumption $n \sigma / \sqrt{2} >1$ in order to apply Lemma~\ref{lemma:erfinequality}.
\end{proof}

\subsection{Combination of Both Approaches}
For a given error threshold $\tilde{e}>0$ and a given $\sigma>0$, we want to obtain the lowest possible $n$ that guarantees that the error threshold is not exceeded. Solving the bound from Proposition~\ref{prop:f} for $n$ and taking the precondition for $n$ into account yields
\begin{align*}
n \geq  1 + \frac{\sigma}{\pi} \sqrt{- \log(4 \pi^3 \tilde{e}^2)} \quad \land \quad n> 1 + \frac{\sigma}{\sqrt{2} \pi}  \ .
\end{align*}
By applying the method to the results of Proposition~\ref{prop:g}, we obtain
\begin{align*}
n \geq \frac{1}{\sigma} \sqrt{- \log (2 \pi^2 \sigma^2 \tilde{e}^2)} \quad \land \quad n > \frac{\sqrt{2}}{\sigma} \ .
\end{align*}
Thus, we define
\begin{align*}
n_f &:= \max \left( 1 + \frac{\sigma}{\pi} \sqrt{- \log(4 \pi^3 \tilde{e}^2)}\ ,\ 1 + \frac{\sigma}{\sqrt{2} \pi} \right) \ , \\
n_g &:= \max \left( \frac{1}{\sigma} \sqrt{- \log (2 \pi^2 \sigma^2 \tilde{e}^2)} \ , \ \frac{\sqrt{2}}{\sigma} \right) \ . 
\end{align*}
Consequently, we set $n := \lceil \min(n_f, n_g) \rceil$ and choose the according method for approximation. Examples with $\tilde{e}=1E-5$ and $\tilde{e}=1E-15$ are given in Fig.~\ref{fig:bounds}. Note that the required $n$ is slightly higher than the empirically obtained values given in Table~\ref{table:combineapprox}, because the theoretical bounds are not tight.

\begin{figure}
	\centering
	\includegraphics[width=8cm]{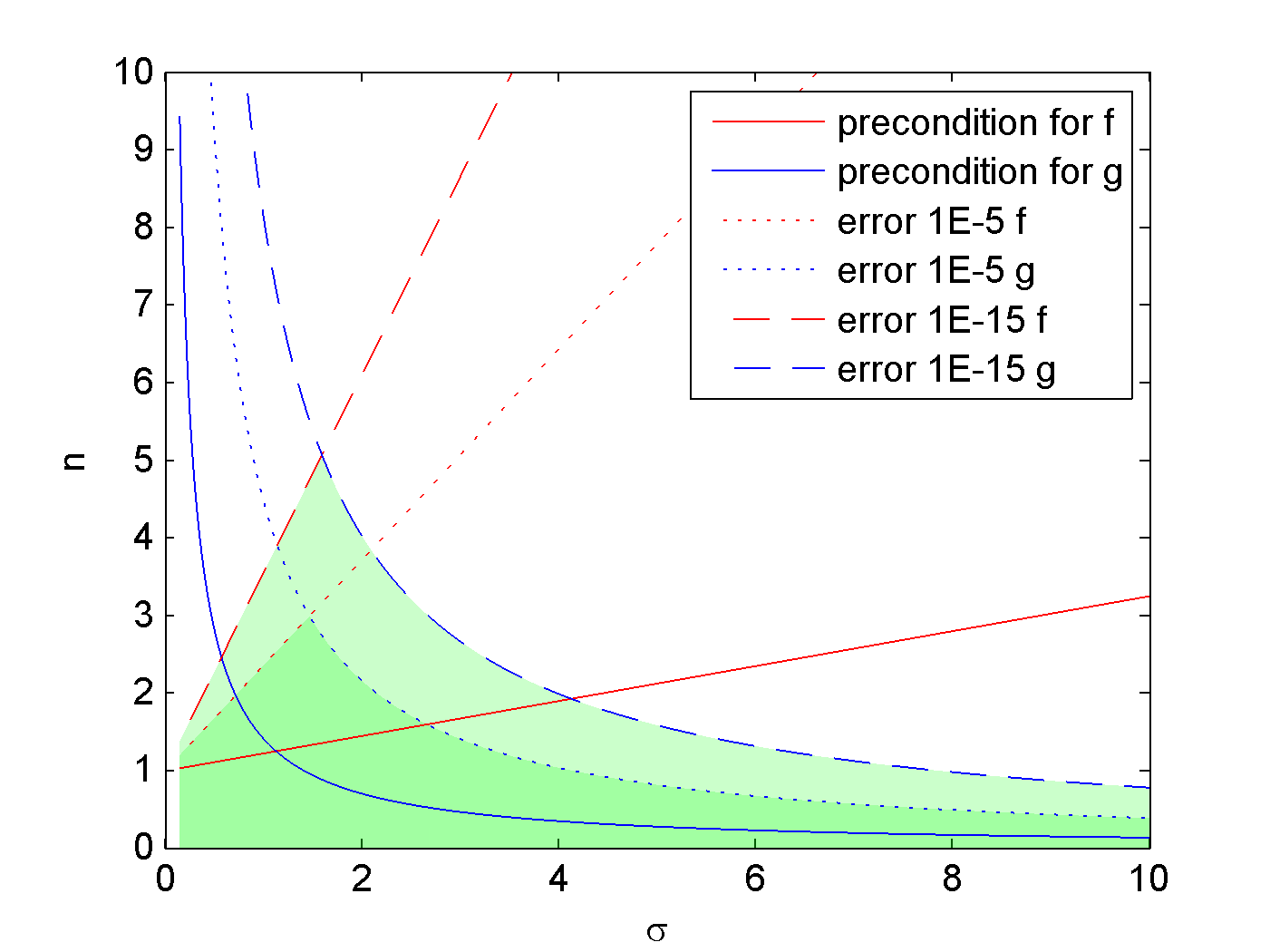}
	\caption{Theoretical results for minimum value of $n$. We consider $\tilde{e}=1E-5$ and $\tilde{e}=1E-15$. The required $n$ by combining both approximations is shaded in dark green and light green respectively.}
	\label{fig:bounds}
\end{figure}

\section{Conclusion} \label{sec:conclusion}
In this paper, we have shown theoretical bounds on two different representations of the wrapped normal probability density function based on truncated infinite series. In both cases, the error decreases exponentially with increasing number of summands $n$. Furthermore, we have shown that one representation performs well for small $\sigma$ whereas the other performs well for large $\sigma$. This motivates their combined use depending on the value of $\sigma$. Our empirical results match well with the theoretical conclusions. We have proposed piecewise approximations based on the two representations with a varying number of summands for several levels of accuracy. 

\section*{Acknowledgment} \noindent
This work was partially supported by grants from the German Research Foundation
(DFG) within the Research Training Groups RTG 1194 ``Self-organizing
Sensor-Actuator-Networks'' and RTG 1126 ``Soft-tissue Surgery: New
Computer-based Methods for the Future Workplace''.



\bibliographystyle{IEEEtran_Capitalize}
\bibliography{../../../Literatur/gk,../../ISASPublikationen/BibTex/ISASPublikationen_laufend}


\end{document}